\documentclass[letterpaper,10 pt, conference]{ieeeconf}
\IEEEoverridecommandlockouts                         
\usepackage{fancyhdr}
\usepackage{hyperref}
\usepackage{xspace}
\usepackage[font={small}]{caption}
\usepackage{color}
\usepackage{setspace}
\usepackage{enumerate}
\usepackage{graphics}
\usepackage{cite}
\usepackage{dsfont}
\usepackage{latexsym}
\usepackage{amsmath}
\usepackage{amssymb}
\usepackage{amsfonts}
\usepackage{amsthm}
\usepackage{times}
\usepackage{color}
\usepackage{verbatim}
\usepackage{xcolor}
\usepackage{epsfig} 
\usepackage{subcaption}
\usepackage{booktabs}
\usepackage{multirow,array}

\let\labelindent\relax
\usepackage{enumitem}

\def\mbb{\mathbb}
\def\mcal{\mathcal}

\def\dSP0{\delta_{SP0}}
\def\dRT0{\delta_{RT0}}
\def\dPS1{\delta_{PS1}}
\def\dTR1{\delta_{TR1}}

\def\dgdn{\frac{\partial g}{\partial n}}

\def\int{ {\text{int}} }

\def\balph{\bar{\alpha}}

\def \be {\begin{equation}}
\def \ee {\end{equation}}
\def \ba {\begin{aligned}}
\def \ea {\end{aligned}}

\theoremstyle{plain}
\newtheorem{theorem}{Theorem}[section]

\newtheorem{proposition}{Proposition}[section]
\newtheorem{lemma}{Lemma}[section]

\newtheorem{assumption}{Assumption}

\newtheorem{definition}{Definition}

\def\bs{\boldsymbol}

\normalsize\title{\LARGE \bf
	The Tragedy of the Commons in Multi-Population Resource Games
	\thanks{ }}

\author{
	Yamin Vahmian, Keith Paarporn
	\thanks{Y. Vahmian and K. Paarporn are with the Department of Computer Science, University of Colorado, Colorado Springs. Contact: \texttt{ \{yvahmian,kpaarpor\}@uccs.edu}. This paper is based on previous work that studied a two-population setting \cite{paarporn2024two}, but with no strategic game interactions between high-level decision-makers. This work is supported by ONR grant \#N000142612120, and in part by NSF grant \#ECCS-2346791.
	}
}

\begin{document}
\thispagestyle{plain}
\pagestyle{plain}

\maketitle

\begin{abstract}	
	Self-optimizing behaviors can lead to outcomes where collective benefits are ultimately destroyed, a well-known phenomenon known as the ``tragedy of the commons". 
	These scenarios are widely studied using game-theoretic approaches to analyze strategic agent decision-making. 
	In this paper, we examine this phenomenon in a bi-level decision-making hierarchy, where low-level agents belong to multiple distinct populations, and high-level agents make decisions that impact the choices of the local populations they represent. 
	We study strategic interactions in a context where the populations benefit from a common environmental resource that degrades with higher extractive efforts made by high-level agents. 
	We characterize a unique symmetric Nash equilibrium in the high-level game, and investigate its consequences on the common resource. 
	While the equilibrium resource level degrades as the number of populations grows large, there are instances where it does not become depleted.
    We identify such regions, as well as the regions where the resource does deplete.
\end{abstract}

\section{Introduction}

Population games are well-suited for analyzing the emergent collective behaviors among a large number of individuals \cite{sandholm2010population}.
A central topic in this area concerns whether the self-interested behaviors lead to outcomes where common benefits are destroyed.
This is the well-known phenomenon called the ``tragedy of the commons''. A primary research goal is to identify, and ultimately influence, the factors that can avert such tragedies.

Feedback-evolving games are a class of models that directly incorporates an environmental state that co-evolves with population behaviors \cite{weitz2016oscillating,tilman2020evolutionary}.
They have been extensively applied to examine the utilization of common resources~\cite{stella2023impact,zhang2023multi}, social distancing in epidemics~\cite{satapathi2023coupled,saad2023dynamics}, and ecological dynamics~\cite{tilman2020evolutionary}.
This research has extended our understanding of which types of learning dynamics \cite{arefin2021imitation,stella2023impact,paarporn2023non,paarporn2024madness} and control policies \cite{wang2020steering,gong2022limit} are able to avert tragic outcomes. 
However, much of this research is in the context of a single population, where everyone is independently interacting with a common environment.
In many scenarios, individuals' behaviors are influenced by higher-level decision-making agents.
For example, neighboring states or countries set different environmental policies, yet they all consume resources from a shared environment (e.g. fish in the ocean, drinking water, and clean air).

In this paper, we build on this framework by considering multiple populations and a bi-level hierarchy of decision-making. 
Each population is represented by a high-level agent that is able to set the incentive policies for their local population. 
We focus on a particular setting where one ``responsible'' agent implements a pro-environmental policy for its population such that the common resource can be sustained in the absence of other populations.
All other high-level agents are considered ``greedy'', who implement policies that induce total consumption behavior in their populations.
Our study investigates the strategic decision-making among the greedy high-level agents, where they are able to control the allowable magnitude of consumption in their populations.

The goal in this paper then is to characterize the emergent strategic behavior in this hierarchical setting, and identify conditions for which the multi-population behavior induces or averts a tragedy of the commons. 
Specifically, we focus on a normal-form game between the high-level greedy agents with coupling constraints on their strategies, which we call the \emph{resource extraction game}. 
The utilities in this game are based on the stable long-term outcomes from the evolutionary behavior from the low-level individuals.
Each greedy agent seeks to extract as much of the resource as possible -- however, higher extractive efforts degrade the resource level.

Our primary contribution is the characterization of Nash equilibria in the resource extraction game (Theorem~\ref{thm:equilibrium}).
In particular, our results show there is a unique symmetric Nash equilibrium for any instance of the game, i.e. all greedy agents choose the same extractive effort.
We then identify the parameter regimes for which the common resource can be sustained in the equilibrium, even when the number of greedy populations grows large (Proposition~\ref{prop:limit}).


\noindent\textbf{Related works:} There is an emerging body of work that focuses on strategic interactions in multi-population settings.
Evolutionary dynamics over a network of populations is considered in \cite{govaert2022population}, where convergence to equilibria is established.
Externalities between two interacting populations have been studied in contexts such as epidemic spreading \cite{certorio2023epidemic} and shared environmental states \cite{kawano2018evolutionary,betz2024evolutionary}.
Several recent papers propose hierarchical decision-making structures: \cite{chen2025hierarchical} proposes new extensions for population games,
\cite{dayanikli2024multi} studies mean-field games with minor and major agents in the regulation of carbon emissions,
and \cite{datar2022strategic,smartcities8010036} consider competing service providers as high-level agents that strategically set prices.

Section \ref{sec:prelim} provides preliminaries on single-population feedback-evolving games. Section \ref{sec:model} formulates our multi-population setup and resource extraction game.
Section \ref{sec:analysis} presents an analysis and main results. Section \ref{sec:quality} examines consequences of the main results with numerical studies.

\section{Preliminary: single population}\label{sec:prelim}


Consider a single population of interacting agents that have access to a common pool resource, where their choices are to cooperate (low consumption, L) or defect (high consumption, H). At time $t \geq 0$, let $x(t) \in [0,1]$ denote the fraction of low consumers and $n(t) \in [0,1]$ the resource level, with $n=1$ corresponding to a fully abundant state and $n=0$ to a depleted state. In a standard feedback-evolving game, these quantities obey the following coupled dynamics:
\begin{equation}\label{eq:1pop}
    \begin{aligned}
        \dot{x} &= x(1-x)g(x,n) \\
        \dot{n} &=  n(1-n)(\theta x - \alpha(1-x))
    \end{aligned}
\end{equation}
where $g(x,n) := \pi_L(x,n)-\pi_H(x,n)$ is the payoff difference between low and high consumers where $\pi_L(x,n)$, $\pi_H(x,n)$ denote their payoffs respectively (to be defined soon), $\theta>0$ is the restoration rate due to low consumption, and $\alpha > 0$ is the resource extraction rate due to high consumption activity. 
The $\dot{x}$ equation above is the standard replicator dynamics, and the $\dot{n}$ equation describes logistic growth for the resource level. This system exhibits a feedback mechanism: agents' strategies influence $n$ which in turn affects the payoffs. We note that the system \eqref{eq:1pop} is forward-invariant in the interior set, $(x,n) \in (0,1)^2$.

The agent payoffs are described by the environment-dependent $2\times 2$ payoff matrix,
\begin{equation}
	A_n = n\begin{bmatrix} R_1 & S_1 \\ T_1 & P_1 \end{bmatrix} + (1-n)\begin{bmatrix} R_0 & S_0 \\ T_0 & P_0 \end{bmatrix}
\end{equation}
Here, the first row and column corresponds to a low consumer, and the second row and column corresponds to a high consumer. Entry $ij$ ($i,j \in \{\mcal{L},\mcal{H}\}$) indicates the experienced payoff to an agent using strategy $i$ when encountering an agent using strategy $j$. We denote $x \in [0,1]$ as the fraction (or frequency) of agents in the population using strategy $\mcal{L}$. The payoff experienced by each type of agent is then given by
\begin{equation}\label{eq:pi}
	\pi_\mcal{L}(x,n) =  [A_n [x,1-x]^\top ]_1, \quad \pi_\mcal{H}(x,n) = [A_n [x,1-x]^\top ]_2
\end{equation}

The payoffs are determined by the parameters in the $A_0$ and $A_1$ matrices. The $A_1$ matrix is the payoff matrix when the environment is abundant. Following the literature on feedback-evolving games, we make the following assumption about the $A_1$ matrix. 
\begin{assumption}\label{assume:PD}
	High consumption is the dominant strategy in $A_1$, i.e. $\dTR1 := T_1 - R_1 > 0$ and $\dPS1 := P_1 - S_1 > 0$
\end{assumption}
On the other hand, the $A_0$ matrix describes incentives when the environment is near depletion, and can be interpreted to be an ``\emph{environmental policy}" that is implemented to reduce pressure on the common resource (e.g. subsidies for using electric vehicles). The payoff structure of the $A_0$ matrix is completely determined from the parameters $\dSP0 := S_0 - P_0$ and $\dRT0 := R_0 - T_0$. 

In this paper, we will also make the following assumption about the parameters $(\dSP0,\dRT0)$.
\begin{assumption}\label{assume:dgdn}
	We assume that $\dgdn(x) < 0$ for all $x \in [0,1]$. This is equivalent to $\dSP0 > -\dPS1$ and $\dRT0 > -\dTR1$.
\end{assumption}
Assumption \ref{assume:dgdn} asserts that the relative payoff to low consumers in the responsible population  monotonically decreases as the resource level improves. In other words, high consumption becomes more incentivized as resource become more available.


The result below summarizes the findings of the originating work \cite{weitz2016oscillating}, which identified conditions on the policies $(\dSP0,\dRT0)$ that enable long-term behavior of \eqref{eq:1pop} to sustain the environmental resource.

\begin{theorem}[adapted from \cite{weitz2016oscillating}]\label{thm:PNAS}
	The environmental policy $(\dSP0,\dRT0)$ determines the asymptotic properties of \eqref{eq:1pop} as follows.
	
	\vspace{1mm}
	
	\noindent 1) \emph{Resource sustained:} If  $(\dSP0,\dRT0)$ satisfies
	\begin{equation}\label{eq:responsible1}
		\dSP0 > 0 \text{ and } -\frac{\theta}{\alpha} \dSP0 \leq \dRT0 < \frac{\dTR1}{\dPS1} \dSP0,
	\end{equation}
	then the fixed point $(x^*,n^*) = (\frac{\alpha}{\alpha + \theta}, \frac{g(x^*,0)}{-\dgdn(x^*)}) \in (0,1)\times[0,1)$ is the only asymptotically stable fixed point in the system. Here, $n^* = 0$ if and only if $-\frac{\theta}{\alpha} \dSP0 = \dRT0$.

    \vspace{1mm}
	
	\noindent 2) \emph{Oscillating tragedy of the commons:} If $\dSP0 > 0$ and $\frac{\dTR1}{\dPS1} \dSP0 < \dRT0$, then the system exhibits convergence to the heteroclinic cycle composed of the boundary of $[0,1]^2$. If $\frac{\dTR1}{\dPS1} \dSP0 = \dRT0$, then all trajectories of the system are closed orbits centered around the neutrally stable interior fixed point $(x^*,n^*) = (\frac{\alpha}{\alpha + \theta}, \frac{g(x^*,0)}{-\dgdn(x^*)}) \in (0,1)^2$.

    \vspace{1mm}
	
	\noindent 3) \emph{Resource collapse:} If $(\dSP0,\dRT0)$ does not satisfy the conditions of items 1) or 2), then the only asymptotically stable fixed point has $n = 0$.
    
\end{theorem}
The policies described in item 1 above induces a stable and non-zero resource level, and the policies in items 2 and 3 induce undesirable environmental outcomes -- either the resource oscillates between bad and good states indefinitely, or it collapses entirely. 
We consequently obtain a natural definition for a set of ``responsible" policies.

\begin{figure}
    \centering
    \includegraphics[width=0.75\linewidth]{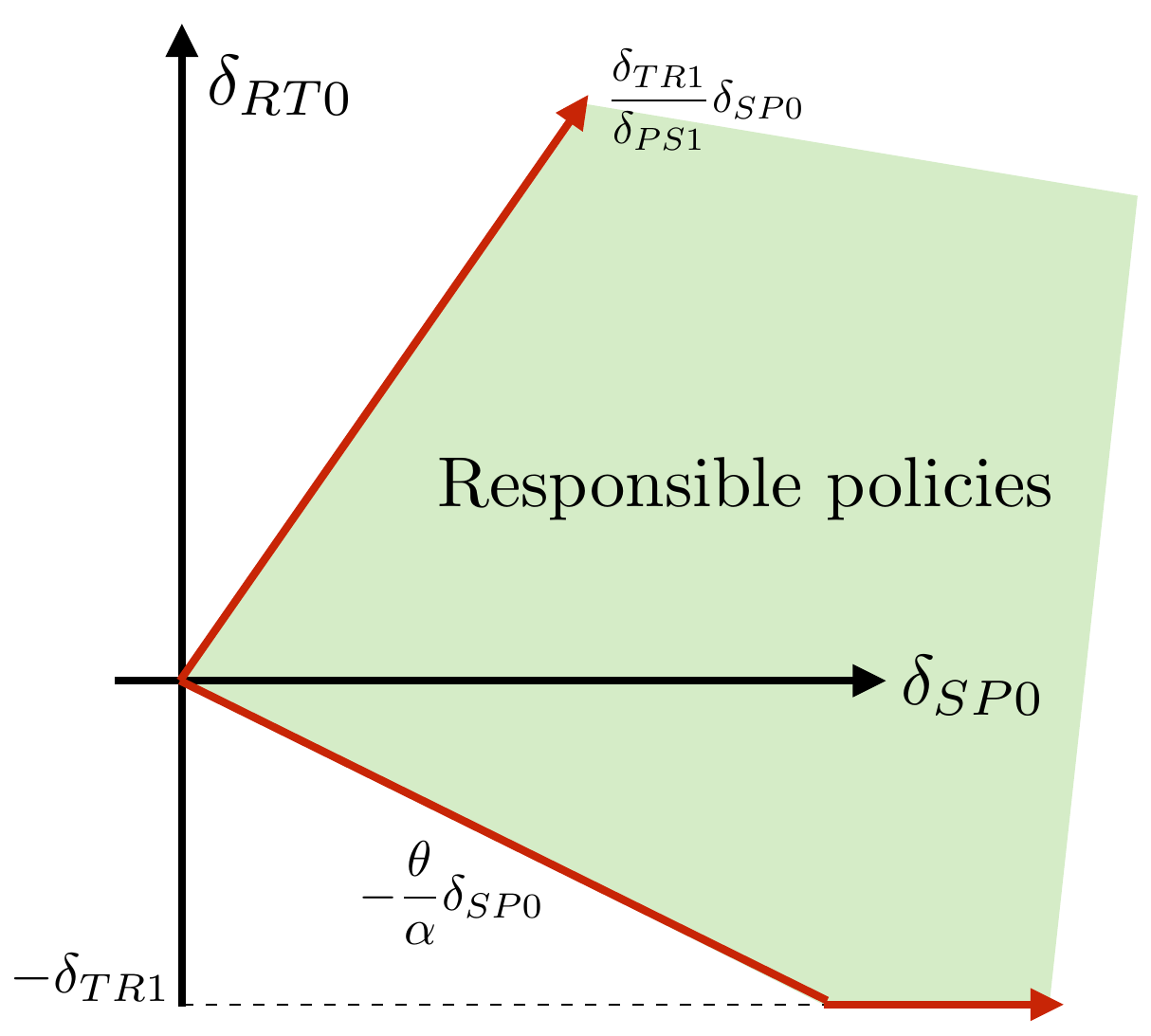}
    \caption{The shaded green region depicts the set of all policies $(\dSP0,\dRT0)$ that are responsible, as defined as in Definition \ref{def:responsible}.}\label{fig:responsible}
\end{figure}

\begin{definition}\label{def:responsible}
    An environmental policy $(\dSP0,\dRT0)$ is \emph{responsible} if it satisfies~\eqref{eq:responsible1} and Assumption \ref{assume:dgdn}. Specifically,
    \begin{equation}
        \max\{\frac{-\theta}{\alpha}\dSP0,-\dTR1\} \leq \dRT0 < \frac{\dTR1}{\dPS1} \dSP0.
    \end{equation}
\end{definition}
In other words, a responsible policy for a population averts a ``tragedy of the commons" (TOC), which we refer to as instances where $n(t) \rightarrow 0$. 
Figure~\ref{fig:responsible} shows a diagram of the region of responsible policies.

\section{Model: Extraction game with multiple populations}\label{sec:model}

Consider an extension of \eqref{eq:1pop} to multiple populations, all of whom share the same common-pool resource. 
In this paper, we focus on a scenario where one of the populations has a responsible policy, and all other populations do not, which are termed ``greedy". 
The state of this system is specified as $(x, x_1, \ldots,x_M,n) \in [0,1]^{M+2}$, where $x$ denotes the fraction of cooperators in the responsible population, $x_i$ the fraction of cooperators in greedy population $i \in \mcal{G} := \{1,\ldots,M\}$, and $n$ the common resource level. 
The state evolves according to the dynamics
\begin{equation}\label{eq:multipop}
    \begin{aligned}
        \dot{x} &= x(1-x)g(x,n) \\
        \dot{x}_i &= x_i(1-x_i)g_i(x,n), \ \forall i \in\mcal{G} \\
        \dot{n} &= \epsilon n(1-n)\left( \theta x + \sum_{i \in \mcal{G}} \theta_i x_i \right. \\
        &\quad\quad\quad \left. - \left( \alpha(1-x) + \sum_{i \in \mcal{G}} \alpha_i(1-x_i) \right) \right).
    \end{aligned}
\end{equation}
Here, we have introduced corresponding rate parameters $\alpha_i,\theta_i > 0$ for the greedy populations. In line with Assumption \ref{assume:PD}, we also assume corresponding payoff parameters in the abundant state satisfy $\dTR1^i, \dPS1^i > 0$.  However, for all greedy populations, we will assume that their environmental policies all satisfy $\dSP0^i,\dRT0^i < 0$.  This implies that the payoff differences for greedy populations satisfy $g_i(x_i,n) := \pi_L^i - \pi_H^i < 0$ for all states $(x_i,n)$, where $\pi_L^i,\pi_H^i$ are defined analogously to \eqref{eq:pi}. This means that low consumption is never incentivized in greedy populations.

We reserve the un-indexed parameters $\alpha$, $\theta$, $\dTR1,\dPS1$, and $\dSP0,\dRT0$ for those associated with the responsible population. By construction, system \eqref{eq:multipop} is forward-invariant on the set $(0,1)^{M+2}$, and can only admit stable fixed points for which $x_i = 0$ for all $i \in \mcal{G}$. 

\begin{lemma}\label{lem:multistability}
	Denote $\bar{\alpha} := \sum_{i\in\mcal{G}} \alpha_i$. The asymptotic dynamics of the multi-population system  \eqref{eq:multipop} are summarized below.
	
	\noindent 1) Suppose $\bar{\alpha} > \theta$. Then a tragedy of the commons is  asymptotically stable, i.e. $\lim_{t\rightarrow \infty} n(t) = 0$.
	
	\vspace{1mm}
	
	\noindent 2) Suppose $\bar{\alpha} < \theta$. 
	\begin{enumerate}[label=(\alph*)]
		\item If $\frac{\bar{\alpha} - \theta}{\alpha + \bar{\alpha}}\dSP0 \leq  \dRT0 < \frac{\dTR1}{\dPS1}\dSP0$, then the only asymptotically stable fixed point is $(x_1^*,\mathbf{0}_M,n^*)$, where
	\begin{equation}\label{eq:nstar}
		x^* := \frac{\alpha + \bar{\alpha}}{\alpha + \theta}, \quad n^* := -\frac{g(x^*,0)}{\dgdn(x^*)} \in (0,1).
	\end{equation}
		\item If $\dRT0 < \frac{\bar{\alpha} - \theta}{\alpha + \bar{\alpha}}\dSP0$, then a tragedy of the commons is the only asymptotically stable outcome.
	\end{enumerate}
	
	\vspace{1mm}
	
	\noindent 3) Suppose $\balph = \theta$. If $\dRT0 > 0$, then there is a locally stable line segment of fixed points given by 
	\be
		\left\{ (1,\mathbf{0}_M,n) : n \in \left[0,\frac{\dRT0}{\dRT0 + \dTR1}\right) \right\}
	\ee
	All other fixed points are isolated and unstable. If $\dRT0 \leq 0$, then a tragedy of the commons is asymptotically stable.
\end{lemma}
\begin{proof}
    The main arguments for this result are fully detailed in \cite{paarporn2024two}, which established the case for a single greedy population. 
    Because all greedy states $x_i(t) \rightarrow 0$ by construction, the extension to $M$ greedy populations easily generalizes because it only requires using the total consumption rate $\bar{\alpha}$ in place of the consumption rate for the single greedy population. 
    We therefore omit these details, as no new fundamental arguments are needed.
\end{proof}



%

In item 1 above, the greedy populations induce a tragedy of the commons if their total consumption rate, $\bar{\alpha}$, exceeds the responsible population's restoration rate $\theta$. Item 2a provides a region of sustainable policies for population 1. This region gets smaller as $\bar{\alpha}$ increases while remaining less than $\theta$. For notational convenience, let us denote this region as 
\begin{equation}\label{eq:Valpha}
    \small
    \mcal{V}(\bar{\alpha}) := \left\{ (y_1,y_2) : \max\left\{\frac{\bar{\alpha} - \theta}{\alpha + \bar{\alpha}}y_1, -\dTR1 \right\} \leq  y_2 < \frac{\dTR1}{\dPS1}y_1 \right\}
\end{equation}

Observe that $\mcal{V}(\bar{\alpha})$ is a sub-region of the set of responsible policies (Definition~\ref{def:responsible}), in which the lower bound increases as the total extraction $\bar{\alpha}$ increases.
Item 2b provides a region where the responsible population fails to sustain the resource even though $\bar{\alpha} < \theta$. 

\begin{figure}
    \centering
    \includegraphics[width=0.65\linewidth]{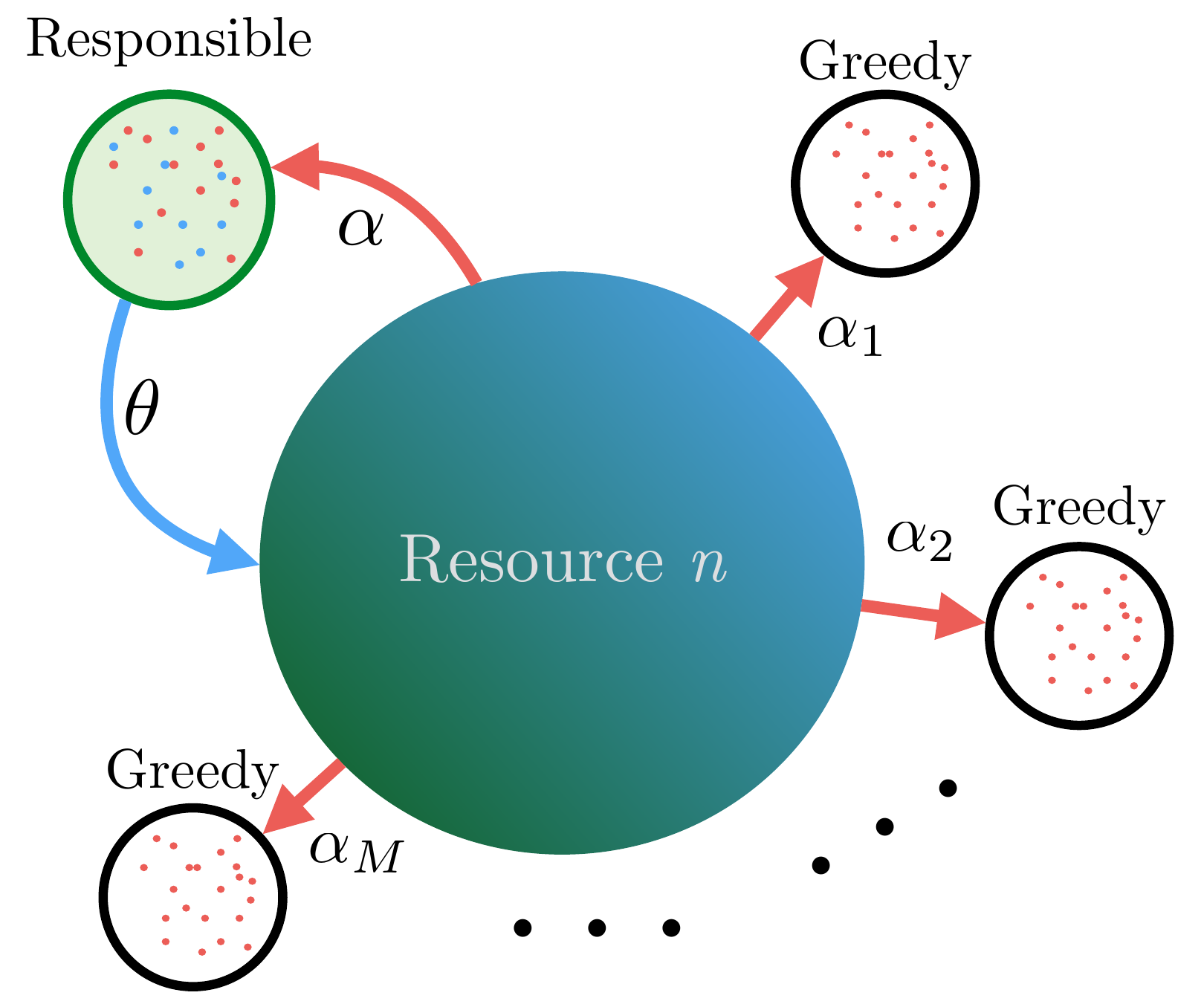}
    \caption{The resource extraction game is a strategic-form game among $M$ players who each represent one of the greedy populations. Player $i\in\mcal{G}$ chooses extraction rate $\alpha_i \geq 0$ for its population, where higher rates degrade the resource.}\label{fig:game}
\end{figure}

\subsection*{Hierarchical resource extraction game}

We now pose a high-level decision-making problem that is the central focus of this paper. We model the resource extraction game as a complete-information game in which each high-level agent observes the aggregate extraction level and the individual extraction choices of the other agents.
This assumption enables agents to compute best responses, and reflects settings where extraction decisions or their aggregate effects are publicly observable or effectively shared among decision-makers.
Consider $M$ high-level agents that represent each of the greedy populations. 
The agents are strategic in selecting its population's resource extraction rate, $\alpha_i \geq 0$. 
This is representative of, for example, the intensity of a corporation's extraction of a natural resource in the presence of competitors, as well as analogous settings such as competing jurisdictions regulating shared environmental resources or firms influencing common infrastructure capacity. 
The \emph{utility} of each agent $i \in \mcal{G}$ is defined as
\begin{equation}\label{eq:utilities}
    U_i(\bs{\alpha}) := \alpha_i R(\bs{\alpha})
\end{equation}
where $\bs{\alpha} := (\alpha_1,\ldots,\alpha_M)$ is the strategy profile of the extraction decisions, and $R(\bs{\alpha})$ is the steady-state resource level induced by the constituent low-level agents, characterized in Lemma~\ref{lem:multistability}:
\begin{equation}\label{eq:resource_level}
    R(\balph) :=
    \begin{cases}
        -\frac{g(\frac{\alpha+\balph}{\alpha+\theta},0)}{\dgdn(\frac{\alpha+\balph}{\alpha+\theta})}, \text{ if } \bar{\alpha} \leq \theta \text{ and } (\dSP0,\dRT0) \in \mcal{V}(\bar{\alpha}) \\
        0, \text{ else}
    \end{cases}
\end{equation}
Lemma \ref{lem:multistability} (item 3) states that when $\bar{\alpha} = \theta$, there is a range of possible asymptotic resource levels. In our definition of $R(\cdot)$ above, we have elected to assign the highest stable resource level
\footnote{This is done for two reasons. First, it makes the resource function well-defined and left-continuous at $\theta$. Thus, $R$ attains a maximum value in the interval $[0, \theta]$. Second, since no particular initial condition is prescribed in the model, we may view the irresponsible populations to be ``optimistic" regarding the best case among all possible outcomes.}, $\frac{\dRT0}{\dRT0 + \dTR1}$.

The utilities~\eqref{eq:utilities} define a strategic-form game between the $M$ local authorities, with continuous strategy spaces $\mcal{A}_i := \mbb{R}_{\geq 0}$. 
Figure~\ref{fig:game} provides an illustration. 
We refer to this game as the \emph{resource extraction game}, and denote an instance of it as the tuple $\mcal{E}(\dSP0,\dRT0,\alpha,\theta)$ of fixed parameters. Table \ref{tab:variables} summarizes the key system parameters and decision variables in the resource extraction game.

In our forthcoming analysis, we seek to characterize Nash equilibria of $\mcal{E}$, i.e. any strategy profile $\bs{\alpha}$ that satisfies $U_i(\alpha_i,\alpha_{-i}) \geq U_i(\alpha_i',\alpha_{-i})$ for all $i \in \mcal{G}$ and for all $\alpha_i' \neq \alpha_i$. 
In particular, we wish to draw attention to the extractive behavior and resource level $R(\bs{\alpha})$ under equilibria of the game, and how these quantities change as the number of greedy populations increases.

\begin{table}[htbp]
\centering
\caption{Parameters and variables of the multi-population system}
\label{tab:variables}
\small
\begin{tabular}{@{}p{1.4cm}p{6.5cm}@{}}
\toprule
\multicolumn{2}{@{}l@{}}{\textbf{System Parameters}} \\
\midrule
$\dSP0,\dRT0$ & Responsible policy in deplete state (see Figure~\ref{fig:responsible}) \\
$\dTR1, \dPS1$ & Responsible policy in replete state (fixed, both positive) \\
$\alpha, \theta$ & Consumption, restoration rates of responsible population \\
\midrule
\multicolumn{2}{@{}l@{}}{\textbf{Decision Variables}} \\
\midrule
$\alpha_i$ & Consumption rate of greedy population $i =1,\ldots,M$ \\
$\balph$ & Total consumption rate of greedy populations \\
$\balph_{-i}$ & Total consumption rate of greedy populations, excluding population $i$ \\
\bottomrule
\end{tabular}
\end{table}

\section{Analysis and Main Results}\label{sec:analysis}

In this section, we seek to characterize Nash Equilibria of the extraction game $\mcal{E}$. Notice from \eqref{eq:resource_level} that the shared resource will be destroyed if the total extraction $\bar{\alpha} > \theta$, upon where all players will receive zero utility. Indeed, there are an infinite number of equilibria under which  the resource is destroyed.
\begin{lemma}
    An equilibrium $\bs{\alpha}$ of $\mcal{E}$ satisfies $R(\bs{\alpha}) = 0$ if and only if for all $i \in \mcal{G}$, we have $\bar{\alpha} - \alpha_i > \theta$.
\end{lemma}
Under the equilibria described above, the total extraction will still exceed $\theta$ regardless of any player's unilateral deviation. 
We thus focus our attention on finding equilibria for which the resource level is sustained, i.e. $R(\bs{\alpha}) > 0$. 
Moving forward, we consider extraction profiles that do not destroy the resource, i.e. ones that obey the following coupling constraint,
\begin{equation}
    \bs{\alpha} \in \mcal{C} := \{\bs{\alpha} : \bar{\alpha} \leq \theta \text{ and } (\dSP0,\dRT0) \in \mcal{V}(\bar{\alpha})\}.
\end{equation}
This causes player $i$'s available extraction choices to be coupled with the profile of extraction rates $\alpha_{-i}$ from the other players. Let us define a player's restricted strategy set as
\begin{equation}\label{eq:restricted_strategies}
    \begin{aligned}
        \mcal{A}_i(\alpha_{-i}) &:= \\
        &\hspace{-8mm}\begin{cases}
            [0,\theta-\bar{\alpha}_{-i}], &\text{ if } \dRT0 > 0 \\
            [0,\frac{\alpha\dRT0 + \theta\dSP0}{\dSP0-\dRT0} - \bar{\alpha}_{-}i], &\text{ if } \frac{\bar{\alpha}_{-i} - \theta}{\alpha + \bar{\alpha}_{-i}}\dSP0 \leq \dRT0 \leq 0
        \end{cases}.
    \end{aligned}
\end{equation}
Let us now denote $\mcal{E}^*_M(\dSP0,\dRT0,\alpha,\theta)$ to mean the extraction game $\mcal{E}$ under the coupling constraint $\mcal{C}$. In the following result, we demonstrate that it is a concave game.

\begin{lemma}\label{lem:concave-game}
    The game $\mcal{E}^*_M(\dSP0,\dRT0,\alpha,\theta)$ is a concave game. That is, for any $i \in \mcal{G}$ and any $\alpha_{-i}$, the function $U_i(\alpha_i,\alpha_{-i})$ is concave in $\alpha_i$.
\end{lemma}
\begin{proof}
    To establish concavity, we need to show that $U_i''(\alpha_i,\alpha_{-i}) < 0$ for all $\alpha_i \in \mcal{A}_i(\alpha_{-i})$ (here, $U'$ denotes partial derivative w.r.t. $\alpha_i$). For all $\alpha_i \in \mcal{A}_i(\alpha_{-i})$, it holds that $R(\alpha_i,\alpha_{-i}) = -\frac{g(x,0)}{\dgdn(x)}$, where $x^*=\frac{\alpha+\bar{\alpha}}{\alpha+\theta}$ \eqref{eq:nstar}. We have
    \begin{equation}
        U''_i(\alpha_i,\alpha_{-i}) = 2 R'(\alpha_i,\alpha_{-i}) + \alpha_i R''(\alpha_i,\alpha_{-i}).
    \end{equation}
    We introduce the following parameters:
    \begin{equation}\label{eq:g_coeffs}
    	\begin{aligned}
        	a &:= \dSP0 - \dRT0 + \dPS1 - \dTR1 \\
        	b &:= \dRT0 - \dSP0 \\
        	c &:= -(\dPS1 + \dSP0) < 0 \\
        	d &:= \dSP0 > 0
    	\end{aligned}
	\end{equation}
    Because these are payoff parameters of the responsible population, we know that $c<0$, $d > 0$, but the signs of $a,b$ can be positive or negative depending on the policy. Using this notation, we can write
    \begin{equation}\label{eq:g_explicit}
        \begin{aligned}
            g(x,n) &= axn + bx + cn + d \\
            \dgdn(x) &= ax + c
        \end{aligned}
    \end{equation}
    for any $x,n \in [0,1]$.
    After algebraic computations, we can write the second derivative in the form
    \begin{equation}
        U''_i(\alpha_i,\alpha_{-i}) = \frac{2Y}{(\alpha+\theta)(\dgdn(x^*))^2}\left(-1 + \frac{\alpha_i}{\alpha+\theta}\frac{a}{\dgdn(x^*)} \right),
    \end{equation}
    where
    \begin{equation}\label{eq:Y_def}
        Y := bc - ad = \dTR1\dSP0 - \dRT0\dPS1 > 0.
    \end{equation}
    Hence, the sign of $U_i''$ is the sign of the term in the parentheses, which we can show is negative. The condition for which this is negative is: 
    \begin{equation}\label{eq:negative_condition}
        \begin{aligned}
            -1 + \frac{a}{\alpha+\theta}\frac{\alpha_i}{\dgdn(x)} < 0 \iff 0 > \dgdn(\frac{\alpha + \bar{\alpha}_{-i}}{\alpha+\theta}).
        \end{aligned}
    \end{equation}
    Here, we have used the fact that $\dgdn(\cdot) < 0$ (Assumption \ref{assume:dgdn}), and that one can write $\dgdn(x) = ax + c$ with $c:=-(\dPS1 + \dSP0)$. This concludes the proof.

\end{proof}

In the next result, we leverage the concavity of the game $\mcal{E}^*_M$ to establish the best-response functions,
\begin{equation}
    \text{BR}_i(\alpha_{-i}) := \arg\max_{\alpha_i\in\mcal{A}_i(\alpha_{-i})} U_i(\alpha_i,\alpha_{-i}).
\end{equation}

\begin{lemma}\label{lem:alpha-star}
    Consider any $\bs{\alpha} \in \mcal{C}$ and any player $i$. Define
    \begin{equation}
        \begin{aligned}
            &F_i(\bar\alpha_{-i}) := \\
            &\frac{\alpha+\theta}{a}\left(-\dgdn\left(\frac{\alpha+\bar{\alpha}_{-i}}{\alpha+\theta}\right) + \frac{1}{b}\sqrt{b\dgdn\left(\frac{\alpha+\bar{\alpha}_{-i}}{\alpha+\theta}\right) Y} \right).
        \end{aligned}
    \end{equation}
     where $a,b$ are defined in \eqref{eq:g_coeffs}, and $Y$ is defined in \eqref{eq:Y_def}. Define
     \begin{equation}
        \begin{aligned}
            &C(\bar\alpha_{-i};\dSP0) := \frac{1}{2}\left[-\left(\dTR1+\frac{\theta-\bar{\alpha}_{-i}}{\alpha+\theta}\dPS1\right) + \right. \\
            &\left. \sqrt{\left(\dTR1+\frac{\theta-\bar{\alpha}_{-i}}{\alpha+\theta}\dPS1\right)^2 + 4\frac{\theta-\bar{\alpha}_{-i}}{\alpha+\theta}\dTR1\dSP0} \right].
        \end{aligned}
    \end{equation}
    Then player $i$'s best response is given as follows. 
    
    \vspace{2mm}
    
    \noindent If $C(\balph_{-i};\dSP0) \leq \dRT0 \leq \frac{\dTR1}{\dPS1}\dSP0$, then
    \begin{equation}\label{eq:BR1}
    	\text{BR}_i(\alpha_{-i}) = \theta-\bar{\alpha}_{-i}.
    \end{equation}

    \noindent If  $\max\{\frac{\bar{\alpha}_{-i} - \theta}{\alpha - \bar{\alpha}_{-i}}\dSP0,-\dTR1\} \leq \dRT0 < C(\balph_{-i};\dSP0)$, then
    \begin{equation}\label{eq:BR2}
        \text{BR}_i(\alpha_{-i}) =
        \begin{cases}
            F_i(\bar\alpha_{-i}), &\text{if } a \neq 0 \\
            \frac{1}{2}\left(\frac{\theta\dSP0 + \alpha\dRT0}{\dSP0 - \dRT0} - \bar{\alpha}_{-i} \right), &\text{if } a = 0
        \end{cases}.
    \end{equation}

\end{lemma}
\begin{proof}
    From Lemma \ref{lem:concave-game}, the function $D(\alpha_i) = U_i(\alpha_i,\alpha_{-i})$ is concave and continuous in its restricted strategy set $\alpha_i \in \mcal{A}_i(\alpha_{-i})$ \eqref{eq:restricted_strategies}.
    \begin{equation}
        \begin{aligned}
            D'(\alpha_i) &= R(\bs{\alpha}) + \alpha_i R'(\alpha_i,\alpha_{-i}) \\
            &= -\frac{g(x,0)}{\dgdn(x)} - \alpha_i \frac{Y}{(\alpha+\theta) (\dgdn(x))^2}.
        \end{aligned}
    \end{equation}
    where $x = \frac{\alpha+\alpha_i + \sum_{j\neq i} \alpha_j }{\alpha+\theta}$. We first verify that it is increasing at $\alpha_i = 0$. 
    Writing $\hat{\alpha}_{-i} = \frac{\alpha+\bar{\alpha}_{-i}}{\alpha+\theta}$ for compactness, and invoking Assumption \ref{assume:dgdn}, 
    the condition $D'(0) > 0$ is equivalent to the condition $g(\hat{\alpha}_{-i},0) = \dSP0(\frac{\theta-\bar{\alpha}_{-i}}{\alpha+\theta}) + \dRT0(\frac{\alpha+\bar{\alpha}_{-i}}{\alpha+\theta}) > 0$. 
    From the constraint $(\dSP0,\dRT0) \in \mcal{V}(\bar{\alpha}_i)$ \eqref{eq:Valpha}, we use the fact that $\dRT0 > \frac{\bar{\alpha}_{-i}-\theta}{\alpha+\theta} \dSP0$ to obtain $g(\hat{\alpha}_{-i},0) > 0$.

    \noindent\underline{\textbf{Case 1:} $C(\bar\alpha_{-i};\dSP0) \leq \dRT0 \leq \frac{\dTR1}{\dPS1}\dSP0$}. 
    In this regime, it holds that $\dRT0 > 0$ so that $\mcal{A}_i(\alpha_{-i}) = [0,\theta-\bar{\alpha}_{-i}]$. 
    The expression $D'(\theta-\bar{\alpha}_{-i})$ is a convex quadratic function in $\dRT0$, whose only positive root is given precisely by $C(\bar\alpha_{-i};\dSP0)$. 
    Therefore for all $\dRT0 \geq C_i(\bar\alpha_{-i};\dSP0)$, $D'(\theta-\bar{\alpha}_{-i}) \geq 0$. 
    Because $D'(0) > 0$, this means that $D$ is monotonically increasing on $\alpha_i \in \mcal{A}_i(\alpha_{-i})$, so that $\text{BR}_i(\alpha_{-i}) = \theta-\bar{\alpha}_{-i}$.

    \noindent\underline{\textbf{Case 2a:} $0 < \dRT0 < C(\bar\alpha_{-i};\dSP0)$}. Here, $\mcal{A}_i(\alpha_{-i}) = [0,\theta-\bar{\alpha}_i]$, and $D$ must have a critical point in $(0,\theta-\bar{\alpha}_i)$ because $D'(\theta-\bar{\alpha}_i) < 0$. The solution to the equation $D'(\alpha_i) = 0$ satisfies:
    \begin{equation}
        \dgdn(x)g(x,0) + Y \frac{\alpha_i}{\alpha+\theta} = 0 
    \end{equation}
    Using \eqref{eq:g_explicit}, we can write $\dgdn(x)g(x,0) = (a\frac{\alpha_i}{\alpha+\theta} + \dgdn(\hat{\alpha}_{-i}))(b\frac{\alpha_i}{\alpha+\theta} + g(\hat{\alpha}_{-i},0))$. The equation $D'(\alpha_i) = 0$ then becomes
    \begin{equation}\label{eq:BR_quadratic_eqn}
        \small
        ab\left(\frac{\alpha_i}{\alpha+\theta} \right)^2 + 2b\dgdn\left(\hat{\alpha}_{-i} \right) \frac{\alpha_i}{\alpha+\theta} + \dgdn\left(\hat{\alpha}_{-i} \right) g\left(\hat{\alpha}_{-i},0\right) = 0.
    \end{equation}
    In the case that $a=0$, then we obtain the solution \eqref{eq:BR2} (second entry). In the case that $a \neq 0$, the roots of this equation can be reduced to be: 
    \begin{equation}\label{eq:root_pm}
        \alpha_i^\pm = \frac{\alpha+\theta}{a}\left(-\dgdn(\hat{\alpha}_{-i}) \pm \frac{1}{b}\sqrt{b\dgdn(\hat{\alpha}_{-i}) Y} \right) = F_i(\bar\alpha_{-i}). 
    \end{equation}
    We are interested only in solutions $\alpha_i \geq 0$. 
    We know $Y$ is a positive constant (proof of Lemma \ref{lem:concave-game}, and $\dgdn(\cdot) < 0$ (Assumption \ref{assume:dgdn}). It also holds here that $b<0$ because $\dRT0 < C_i(\alpha_{-i};\dSP0)$. The sign of the radicand in \eqref{eq:root_pm} is positive, and we have that $\alpha_i^+ \geq 0$. This is because the sign of the parentheses term in \eqref{eq:root_pm} is the same sign as the parameter $a$, thus establishing that $\alpha_i > 0$. To see this, the condition that $-\dgdn(\hat{\alpha}_{-i}) + \frac{1}{b}\sqrt{b\dgdn(\hat{\alpha}_{-i}) Y} > 0$ is equivalent to $ag(\hat{\alpha}_{-i},0) > 0$. Since $g(\hat{\alpha}_{-i},0) > 0$, the condition holds if and only if $a > 0$. 
    
    Now, consider $\alpha_i^-$. Here, the parentheses term is positive because it is the sum of two positive terms. If $a < 0$, then $\alpha_i^- < 0$ and thus $\alpha_i^+ \geq 0$ is the only non-negative root. If $a > 0$, it holds that both roots are non-negative. However, we can deduce that $\alpha_i^+$ is the root that lies in $\mcal{A}_i(\alpha_{-i})$ for the following reasons. First, it holds that  $\alpha_i^+ < \alpha_i^-$. Because $D(\alpha_i)$ is continuous and concave for all $\alpha_i \in \mcal{A}_i(\alpha_{-i})$, it can have at most one critical point on this interval. 
    

    \noindent\underline{\textbf{Case 2b:} $\max\{-\dTR1,\frac{\bar{\alpha}_{-i} - \theta}{\alpha - \bar{\alpha}_{-i}}\dSP0\} \leq \dRT0 \leq 0$}. 
    
    \vspace{2mm}
    
    Here, $\mcal{A}_i(\alpha_{-i}) = [0,\frac{\alpha_r\dRT0 + \theta\dSP0}{\dSP0-\dRT0} - \bar{\alpha}_{-i}]$. Also in this regime, we have that $b < 0$. At $\alpha_i = \frac{\alpha_r\dRT0 + \theta\dSP0}{\dSP0-\dRT0} - \bar{\alpha}_{-i}$, it holds that $g(\frac{\alpha+\bar{\alpha}}{\alpha+\theta},0) = 0$, and thus $D(\frac{\alpha_r\dRT0 + \theta\dSP0}{\dSP0-\dRT0} - \bar{\alpha}_{-i}) = 0$. Since $D'(0) > 0$, there is a critical point in the interior of $\mcal{A}_i(\alpha_{-i})$. This must be given precisely by the expression $\alpha_i^+$ \eqref{eq:root_pm}.

\end{proof}

With the best-response functions established, we now state our main result, which characterizes the equilibrium extraction rate for all instances of $\mcal{E}^*_M$.

\begin{theorem}\label{thm:equilibrium}
    	Define
	\begin{equation}\label{eq:equil_alpha_pm}
		\begin{aligned}
    			E_\pm &:= \frac{1}{2M^2}\left[-\frac{\alpha+\theta}{a}\left(2M\dgdn(\frac{\alpha}{\alpha+\theta}) - \frac{M-1}{b}Y \right) \right. \\
    			& \quad\left. \pm \sqrt{\frac{(\alpha+\theta)^2Y}{a^2b} \left(4M^2 \dgdn(\frac{\alpha}{\alpha+\theta}) + (M-1)^2\frac{Y}{b} \right)} \right].
 		\end{aligned}
	\end{equation}
	The game $\mcal{E}^*_M(\dSP0,\dRT0,\alpha,\theta)$ admits a unique symmetric Nash equilibrium $(\alpha^*,\ldots,\alpha^*)$, given as follows.
	
	\vspace{2mm}
	
	\noindent If $C(\frac{M-1}{M}\theta;\dSP0) \leq \dRT0 \leq \frac{\dTR1}{\dPS1}\dSP0$, then
	\begin{equation}\label{eq:equil1}
		\alpha^* = \frac{\theta}{M}.
	\end{equation}
	\noindent If $\max\{\frac{-\theta}{\alpha}\dSP0,-\dTR1\} \leq \dRT0 < C(\frac{M-1}{M}\theta;\dSP0)$, then
	\begin{equation}\label{eq:equil2}
		\alpha^* = 
        \begin{cases}
            E_+, &\text{if } a < 0 \\
            E_-, &\text{if } a > 0 \\
            \frac{1}{M+1}\frac{\theta\dSP0 + \alpha\dRT0}{\dSP0 - \dRT0}, &\text{if } a = 0
        \end{cases},
	\end{equation}
    where $a$ is defined in \eqref{eq:g_coeffs}.
\end{theorem}
\begin{proof}
	The approach to proving this result is to consider symmetric profiles $(\gamma,\ldots,\gamma)$, and leverage the best-response function from Lemma \ref{lem:alpha-star} by solving the equation $\gamma = \text{BR}((M-1)\gamma)$ under each of its conditions \eqref{eq:BR1}, \eqref{eq:BR2}.

	\noindent\textbf{Case 1:} In order for a symmetric profile $(\gamma,\ldots,\gamma)$ to be an equilibrium and satisfy $C((M-1)\gamma;\dSP0) \leq \dRT0 \leq \frac{\dTR1}{\dPS1}\dSP0$, it must hold from Lemma \ref{lem:alpha-star} that $\gamma = \theta - \bar{\gamma}_{-i} = \theta - (M-1)\gamma$, which yields the unique equilibrium \eqref{eq:equil1}.

	\noindent\textbf{Case 2:}  In order for a symmetric profile $\gamma$ to be an equilibrium and satisfy $\max\{\frac{(M-1)\gamma- \theta}{\alpha - (M-1)\gamma},-\dTR1\} \leq \dRT0 < C((M-1)\gamma;\dSP0)$ (with $a \neq 0$), the equation $\gamma = F((M-1)\gamma)$ must hold. Its solutions satisfy the quadratic equation
	\begin{equation}\label{eq:quadratic_gamma}
		\begin{aligned}
			Q(\gamma) := M^2\gamma^2 + K_1 \gamma + K_0 = 0
 		\end{aligned}
	\end{equation}
	where
	\begin{equation}
		\begin{aligned}
    			K_1 &:= \frac{\alpha+\theta}{a}\left(2M\dgdn(\frac{\alpha}{\alpha+\theta}) - \frac{M-1}{b}Y \right) \\
    			K_0 &:= \frac{(\alpha+\theta)^2}{a^2}\dgdn(\frac{\alpha}{\alpha+\theta}) \left( \dgdn(\frac{\alpha}{\alpha+\theta}) - \frac{Y}{b} \right).
 		\end{aligned}
	\end{equation}
The solutions to this equation are given by \eqref{eq:equil_alpha_pm}. Observe that the radicand is positive, due to $\dgdn(\cdot) < 0$, and $b < 0$ because $\dRT0 < C((M-1)\gamma;\dSP0) < \dSP0$. This means that the solutions are real-valued. 

We can also verify that when $a<0$, we have $K_1 > 0$ and $K_0 < 0$, and when $a>0$, we have $K_1 < 0$ and $K_0 > 0$. In the case that $a<0$, it holds that $Q(0) < 0$ and $Q'(0) > 0$, which implies there is a unique positive solution given by $E_+$.  In the case that $a>0$, it holds that $Q(0) > 0$ and $Q'(0) < 0$, which implies that there are two positive solutions. However, only the solution $E_-$ provides a symmetric profile that is feasible in the restricted strategy sets of all players, i.e. $(E_-,\ldots,E_-) \in \mcal{C}$. 

It holds that $M E_+ < \theta$ (identical argument applies for $E_-$), since otherwise we would be in Case 1. 
It then follows that $C( (M-1)E_+; \dSP0) > C(\frac{M-1}{M}\theta;\dSP0)$, because $C(\cdot;\dSP0)$ is a decreasing function in its first argument. 
We are then able to characterize the equilibria in the first two entries of \eqref{eq:equil2}.

When $a = 0$, the equation $\gamma =  \frac{1}{2}\left(\frac{\theta\dSP0 + \alpha\dRT0}{\dSP0 - \dRT0} - (M-1)\gamma \right)$ must hold, which yields the equilibrium in the last entry of \eqref{eq:equil2}.

\end{proof}

In the symmetric equilibrium $\alpha^*$ of $\mcal{E}^*_M(\dSP0,\dRT0,\alpha,\theta)$, the players choose positive extraction rates, though never large enough to destroy the resource.

The equilibrium characterization in Theorem~\ref{thm:equilibrium}  has two primary regimes with respect to the responsible policy $(\dSP0,\dRT0)$. In the case of \eqref{eq:equil1} where mutual cooperation is highly incentivized (high $\dRT0$), the total extraction $M\alpha^*$ reaches its upper limit $\theta$. Here, we note that the boundary $C(\frac{M-1}{M}\theta;\dSP0)$ is monotonically decreasing in $M$ and reaches 0 in the limit. In the case of \eqref{eq:equil2}, mutual cooperation is not as highly incentivized. In this regime, the total extraction rate is less than the upper limit of $\frac{\alpha\dRT0 + \theta\dSP0}{\dSP0-\dRT0}$.

In the next section, we investigate the impact of equilibrium behavior on the system's resource level, the agents' overall utilities, and the total extraction rate.

\begin{figure*}[t]
  \centering
  
  \begin{minipage}{0.52\textwidth}
    \centering
    \includegraphics[width=1\linewidth]{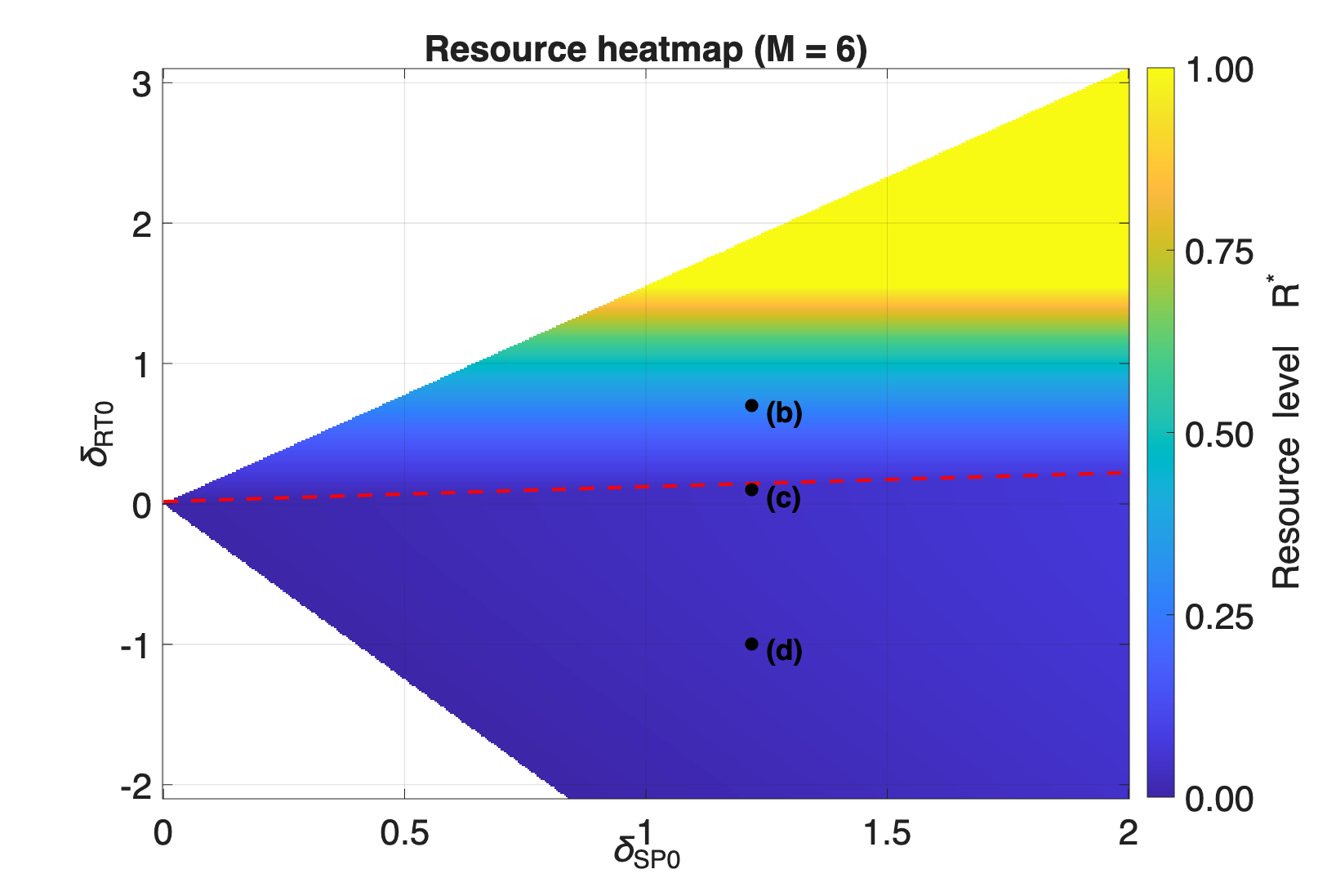}
    \subcaption{}
    \label{fig:big}
  \end{minipage}\hfill
  \begin{minipage}{0.48\textwidth}
    \centering
    \includegraphics[width=0.6\linewidth,height=0.1\textheight]{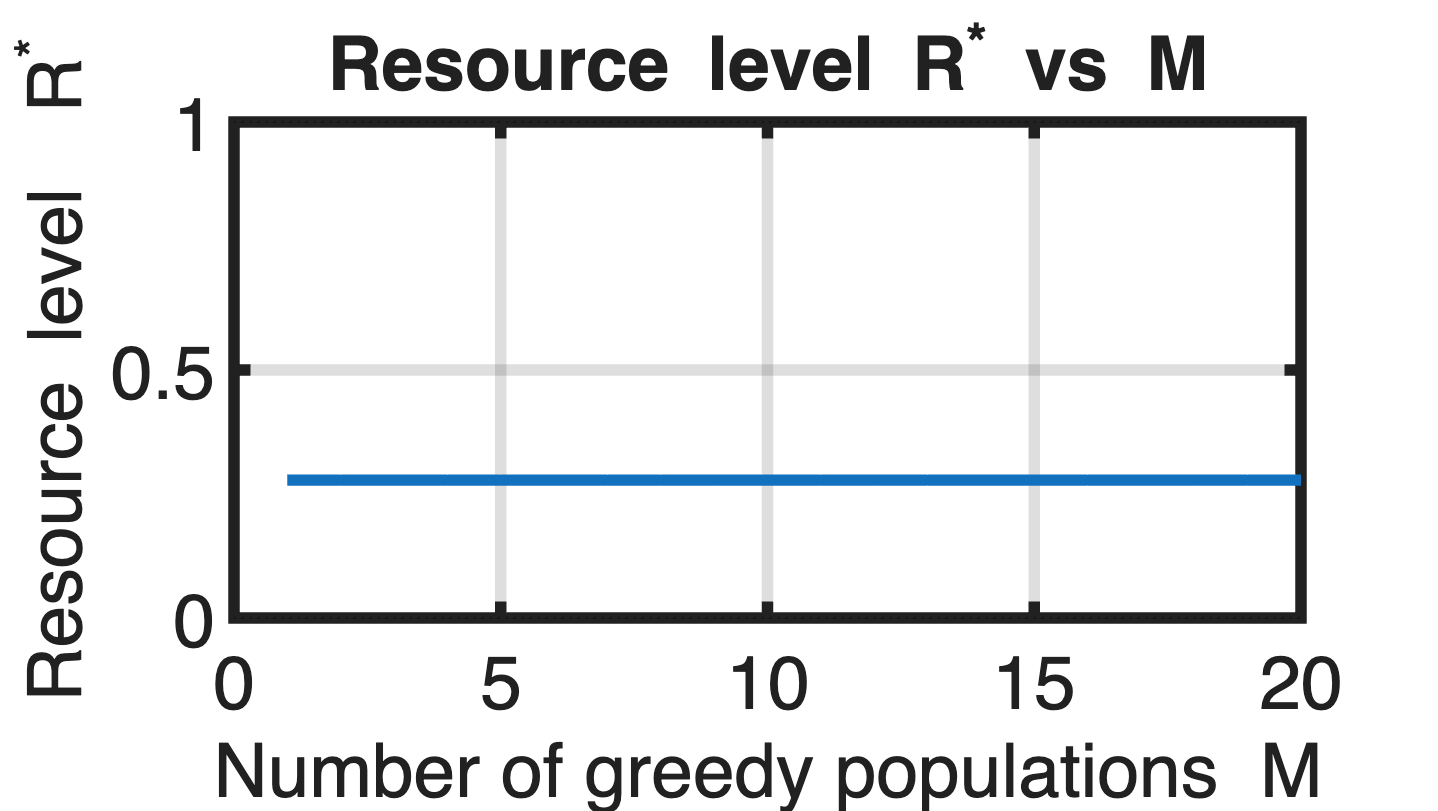}
    \subcaption{}
    
    \vspace{0.06cm}
    \includegraphics[width=0.6\linewidth,height=0.1\textheight]{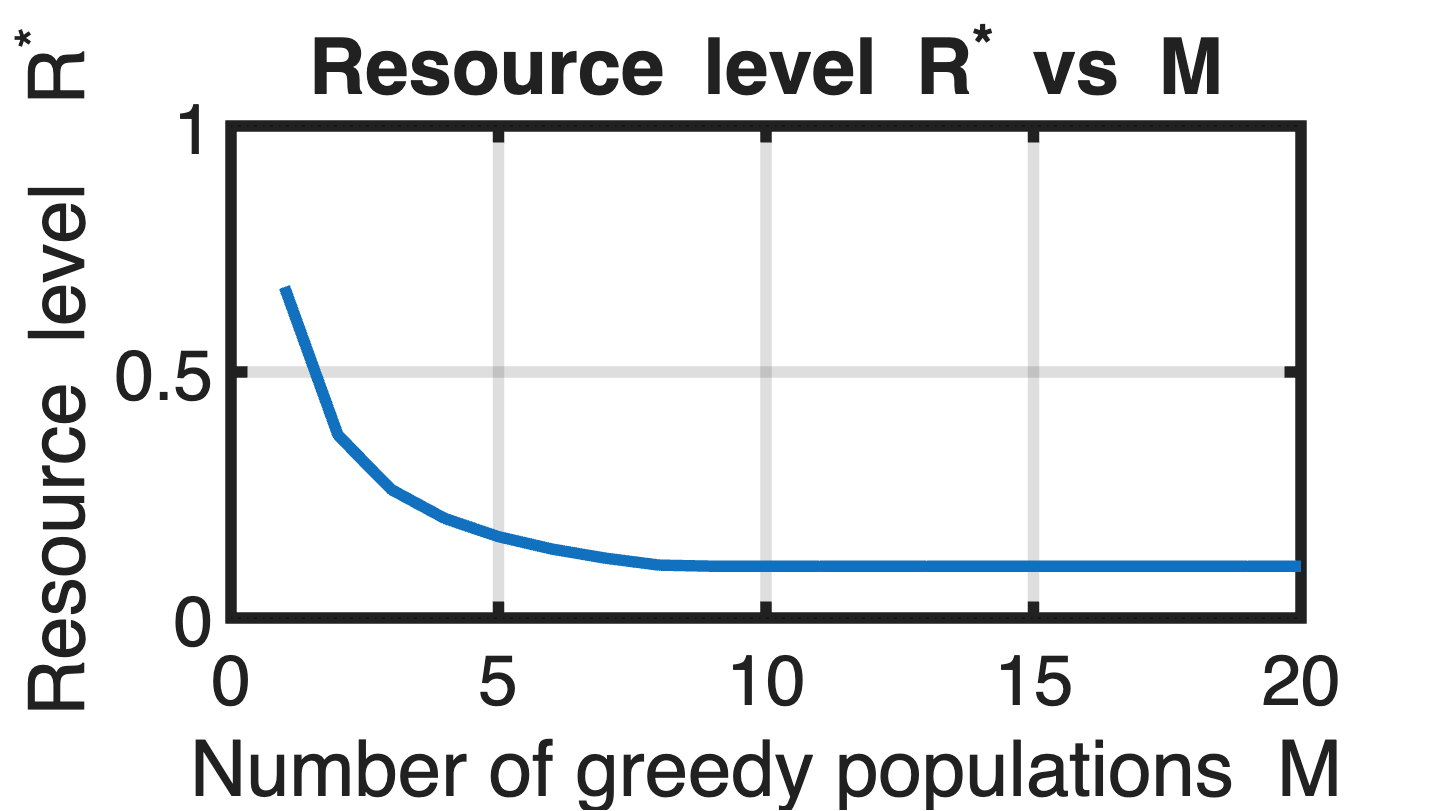}
    \subcaption{}
    
    \vspace{0.06cm}
    \includegraphics[width=0.6\linewidth,height=0.1\textheight]{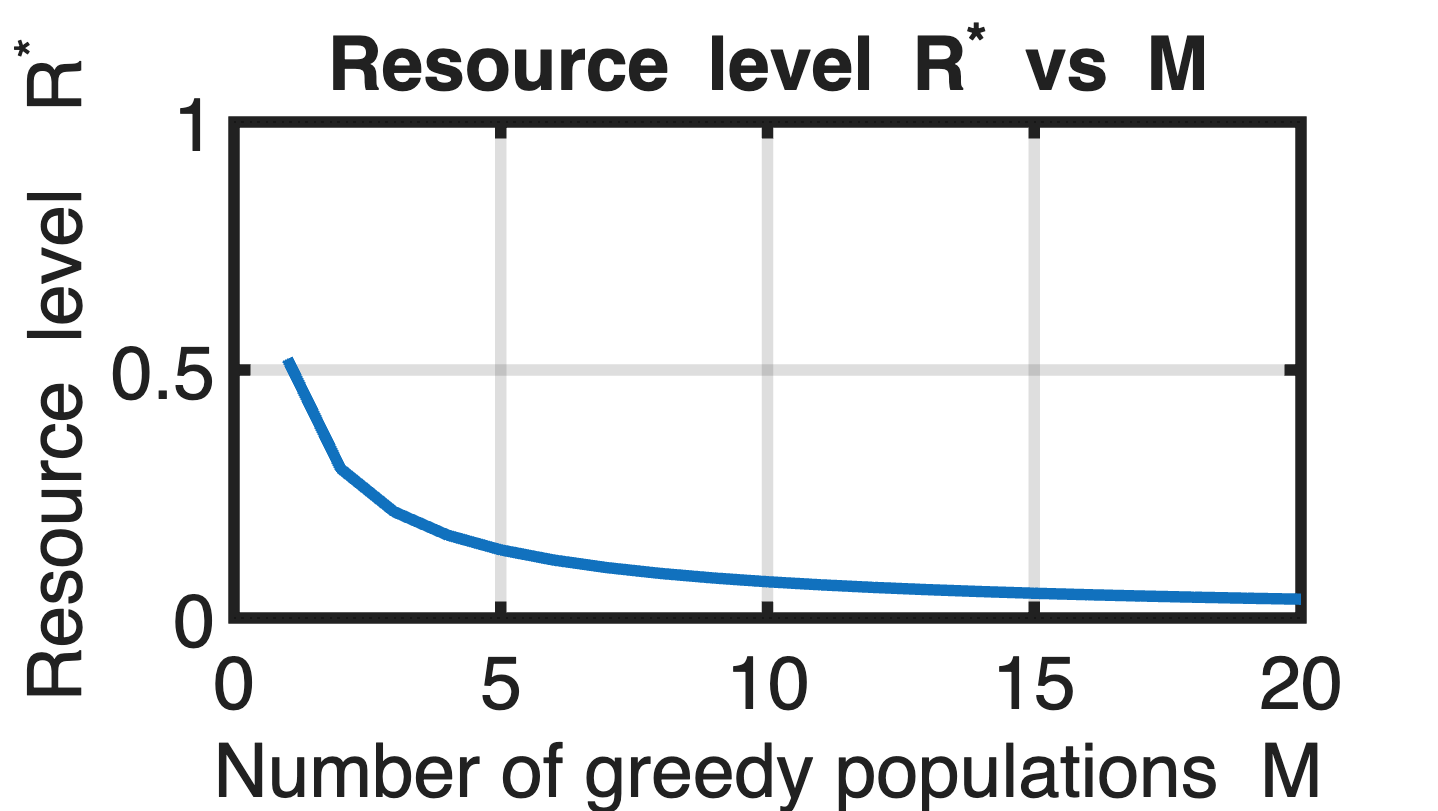}
    \subcaption{}
    \label{fig:smalls}
  \end{minipage}
  \caption{(a) Resource level $R^*_M$ under the symmetric equilibrium (Theorem~\ref{thm:equilibrium}) over the responsible policy space for $M=6$ greedy populations, with parameters $\theta=1.0,\ \alpha=0.40,\ \delta_{TR1}=2.1,\ \delta_{PS1}=2.0$. The red dashed line indicates the curve $C(0;\dSP0)$.
  Figures (b), (c), and (d) show the resource levels $R^*_M$ vs. the number of greedy populations $M$ for selected parameter values (shown as dots in subplot (a)).  
  In (b), $\delta_{RT0}=0.8$ is above $C(0;\dSP0)$, meaning that $\alpha^*=\theta/M$ for all $M \geq 1$, leading to a constant value for $R^*_M$. 
  In (c), $\delta_{RT0}=0.2 > 0$ is below $C(0;\dSP0)$. $R^*_M$ decreases initially then stabilizes at a positive constant (Proposition~\ref{prop:limit}). 
  In (d) $\delta_{RT0}=-1.0 < 0$. $R^*_M$ monotonically decreases toward zero, indicating complete resource depletion.}
  \label{fig:combined}
\end{figure*}

\section{Impact of equilibrium extraction}\label{sec:quality}

We are interested in assessing the quality of the symmetric equilibrium $\alpha^*$ characterized in Theorem \ref{thm:equilibrium}.
We will examine equilibrium outcomes as the number of greedy populations becomes large.
In particular, we are interested in determining in which cases the common resource could be preserved even with a growing number of greedy populations.
Let us denote $\alpha_M^*$ as the equilibrium (individual) extraction rate of any of the $M \geq 1$ greedy populations.
The result below characterizes the limiting values of the equilibrium total consumption rate, $\balph^*_M := M\alpha_M^*$, and the equilibrium resource level, $R^*_M := R(\balph^*_M)$.
\begin{proposition}\label{prop:limit}
	Consider any game $\mcal{E}^*_M(\dSP0,\dRT0,\alpha,\theta)$.
	If $0 < \dRT0 \leq \frac{\dTR1}{\dPS1}\dSP0$, then
	\begin{enumerate}
		\item $\bar{\alpha}^*_\infty := \lim_{M\rightarrow \infty} \balph^*_M = \theta$.
		\item $R^*_\infty := \lim_{M\rightarrow \infty} R^*_M = \frac{\dRT0}{\dRT0+\dTR1} > 0$.
	\end{enumerate}
    If $\max\{\frac{-\theta}{\alpha}\dSP0,-\dTR1\} \leq \dRT0 \leq 0$, then
	\begin{enumerate}
        \setcounter{enumi}{2}
		\item $\bar{\alpha}^*_\infty = \frac{\alpha\dRT0+\theta\dSP0}{\dSP0-\dRT0} \in (0,\theta)$. 
        \item $R^*_\infty = 0$.
	\end{enumerate}
\end{proposition}
\begin{proof}
    In the limit of large $M$, $C(\frac{M-1}{M}\theta;\dSP0) \rightarrow C(\theta;\dSP0) = 0$.
    Thus, if $\dRT0 > 0$, the equilibrium consumption $\alpha^* = \frac{\theta}{M}$ for sufficiently large $M$ \eqref{eq:equil1}.
    We therefore obtain the limits for items 1 and 2.
    
    Now, suppose $\max\{\frac{-\theta}{\alpha}\dSP0,-\dTR1\} \leq \dRT0 \leq 0$, and suppose $a<0$. 
    By \eqref{eq:equil2}, we have $\bar{\alpha}^*(M) = M\cdot E_+$. 
    Applying L'Hopital's rule, we obtain $\bar{\alpha}^*_\infty = -\frac{\alpha+\theta}{a}(\dgdn(\frac{\alpha}{\alpha+\theta}) - \frac{Y}{2b}) + \frac{(\alpha+\theta)Y}{2ab}$. 
    This expression, after substituting $Y = bc-ad$ \eqref{eq:Y_def}, becomes $\frac{\alpha\dRT0+\theta\dSP0}{\dSP0-\dRT0}$, which is the total consumption capacity in this regime.
    Similar calculations can be done when considering $a \leq 0$, by applying the other cases of \eqref{eq:equil2}. This yields item 3.
    To calculate $R^*_\infty = -\frac{g(\frac{\alpha+\bar{\alpha}^*_\infty}{\alpha+\theta},0)}{\dgdn(\frac{\alpha+\bar{\alpha}^*_\infty}{\alpha+\theta})}$, we simply observe that $g(\frac{\alpha+\bar{\alpha}^*_\infty}{\alpha+\theta},0) = b \frac{\dSP0}{\dSP0-\dRT0} + \dSP0 = 0$ 
    (after substituting $b$ from~\eqref{eq:g_coeffs}). This yields item 4. 
    

\end{proof}
The significant observation in Proposition \ref{prop:limit} is that the responsible population's policy $(\dSP0,\dRT0)$ will determine whether the common resource can be sustained, or will become depleted with a growing number of greedy populations.
The multi-population system \emph{averts a tragedy of the commons} in the regime $0 < \dRT0 \leq \frac{\dTR1}{\dPS1}\dSP0$, and \emph{induces a tragedy of the commons} in the regime $\max\{\frac{-\theta}{\alpha}\dSP0,-\dTR1\} \leq \dRT0 \leq 0$.
In both regions, each greedy population's individual utility, $\alpha^* R_M^*$, diminishes to zero as $M$ increases. 
Moreover, the total equilibrium consumption $\balph_M^*$ approaches the resource's capacity~\eqref{eq:restricted_strategies}.

Figure~\ref{fig:combined} provides numerical illustrations of the equilibrium resource levels as they depend on the responsible policy $(\dSP0,\dRT0)$ and the number of populations $M$. Here, we observe three distinct regimes. For large values of $\dRT0$, the total equilibrium consumption is the upper limit $\theta$ for any number of populations $M$, and thus the resource level remains a positive constant value (no tragedy, see Figure~\ref{fig:combined}(b)).  For intermediate values of $\dRT0$, the equilibrium resource level is decreasing in $M$, but it settles to a positive constant value (no tragedy, see Figure~\ref{fig:combined}(c)). For low values of $\dRT0$, the equilibrium resource level is decreasing in $M$ and it converges to zero (a tragedy, see Figure~\ref{fig:combined}(d)).

\section{Conclusion}

In this paper, we investigate a resource extraction game whose players are high-level decision-makers who can determine the consumption rates of the local populations they represent. The players seek to extract as much of the resource for their populations as possible, however, higher total extraction degrades the resource. Our primary results establish a unique symmetric Nash equilibrium. We then identified conditions for which the equilibrium extraction behavior can lead to the destruction of the resource as the number of players increases, a scenario referred to as a tragedy of the commons. The hierarchical extraction game studied here captures real-world settings such as multiple countries, firms, or regional authorities extracting from a shared environmental resource (e.g., fisheries, groundwater basins, or atmospheric carbon sinks), where each decision-maker controls aggregate consumption behavior within its jurisdiction while collectively impacting a common resource.

This work may be extended along a few directions. The uniqueness of the symmetric equilibrium could be rigorously established using the theory of $M$-player concave games. Also, we seek to characterize equilibria in scenarios where populations have different masses. An important direction for future work is to relax the symmetry assumptions by incorporating heterogeneity across populations or considering asymptotic regimes with many players, thereby improving the model’s relevance to large-scale real-world systems.

\bibliographystyle{IEEEtran}
\bibliography{library}

\end{document}